\documentclass[12pt,reqno]{amsart}
\usepackage{amssymb,amsfonts,latexsym,amstext,amsmath,epsfig}
\usepackage[left=3cm,top=3cm,right=3cm,bottom=3cm]{geometry}

\newtheorem{theorem}{Theorem}[section]
\newtheorem{lemma}[theorem]{Lemma}
\numberwithin{equation}{section}

\newcommand{\vol}{{\rm vol}}

\newcommand\eps{\varepsilon}
\newcommand{\E}{\mathbb E}
\newcommand{\Prob}{\mathbb{P}}

\begin{document}
\title{Models of on-line social networks}

\author{Anthony Bonato}
\address{Department of Mathematics\\
Ryerson University\\
Toronto, ON\\
Canada, M5B 2K3} \email{abonato@ryerson.ca}
\author{Noor Hadi}
\address{Department of Mathematics\\
Wilfrid Laurier University\\
Waterloo, NS\\
Canada, N2L 3C5} \email{hadi4130@wlu.ca}
\author{Paul Horn}
\address{Department of Mathematics and Computer Science\\
Emory University\\
Atlanta, GA\\
U.S.A.,  30322} \email{phorn@mathcs.emory.edu}
\author{Pawe\l{} Pra\l{}at}
\address{Department of Mathematics and Statistics\\
Dalhousie University\\
Halifax NS\\
Canada, B3H 3J5} \email{pralat@mathstat.dal.ca}
\author{Changping Wang}
\address{Department of Mathematics\\
Ryerson University\\
Toronto, ON\\
Canada, M5B 2K3}
\email{cpwang@ryerson.ca}

\keywords{complex networks, on-line social networks, transitivity, densification power law, average distance, clustering coefficient, spectral gap, bad expansion, normalized Laplacian}
\thanks{The authors gratefully acknowledge support from NSERC and MITACS}
\subjclass{05C75, 68R10, 91D30}
\maketitle

\begin{abstract}

We present a deterministic model for on-line social networks (OSNs) based
on transitivity and local knowledge in social interactions. In the
Iterated Local Transitivity (ILT) model, at each time-step and for
every existing node $x$, a new node appears which joins to the
closed neighbour set of $x.$ The ILT model provably satisfies a
number of both local and global properties that were observed in
OSNs and other real-world complex networks, such as a
densification power law, decreasing average distance, and higher
clustering than in random graphs with the same average degree.
Experimental studies of social networks demonstrate poor expansion
properties as a consequence of the existence of communities with low
number of inter-community edges. Bounds on the spectral gap for both the
adjacency and normalized Laplacian matrices are proved for graphs
arising from the ILT model, indicating such bad expansion
properties. The cop and domination number are shown to remain the same
as the graph from the initial time-step $G_0$, and the automorphism group
of $G_0$ is a subgroup of the automorphism group of graphs generated at all later time-steps. A randomized version of the ILT model is presented, which
exhibits a tuneable densification power law exponent, and maintains several properties
of the deterministic model.
\end{abstract}

\section{Introduction\label{intro}}

On-line social networks (OSNs) such as Facebook, MySpace, Twitter, and Flickr have
become increasingly popular in recent years. In OSNs, nodes represent
people on-line, and edges correspond to a friendship relation
between them. In these complex real-world networks with sometimes
millions of nodes and edges, new nodes and edges dynamically appear
over time. Parallel with their popularity among the general public
is an increasing interest in the mathematical and general scientific
community on the properties of on-line social networks, in both
gathering data and statistics about the networks, and finding
models simulating their evolution. Data about social interactions in
on-line networks is more readily accessible and measurable than in
off-line social networks, which suggests a need for rigorous models
capturing their evolutionary properties.

The small world property of social networks, introduced by Watts and
Strogatz \cite{sw}, is a central notion in the study of complex
networks, and has roots in the work of Milgram \cite {mil} on short
paths of friends connecting strangers. The small world property
posits low average distance (or diameter) and high clustering, and
has been observed in a wide variety of complex networks.

An increasing number of studies have focused on the small world and
other complex network properties in OSNs. Adamic
et al.\ \cite {ada} provided an early study of an on-line social
network at Stanford University, and found that the network has the
small world property. Correlation between friendship and geographic
location was found by Liben-Nowell et al.\ \cite{ln} using data from
LiveJournal. Kumar et al.\ \cite{kumar2006sae} studied the evolution
of the on-line networks Flickr and Yahoo!360. They found (among other things) that the
average distance between users actually decreases over time, and
that these networks exhibit power-law degree
distributions. Golder et al.\ \cite{golderrsi} analyzed the Facebook
network by studying the messaging pattern between friends with a
sample of 4.2 million users. They also found a power law degree
distribution and the small world property. Similar results were
found in \cite{ahn} which studied Cyworld, MySpace, and Orkut, and
in \cite{mislove2007maa} which examined data collected from four
on-line social networks: Flickr, YouTube, LiveJournal, and Orkut. Power laws for both the in- and out-degree distributions,
low diameter, and high clustering coefficient were reported in the Twitter friendship graph by Java et al.\ \cite{twit0}.
In \cite{twit}, geographic growth patterns and
distinct classes of users were investigated in Twitter.
For further background on complex networks and their models, see the
books \cite{bonato,cald,clbook,dur}.

Recent work by Leskovec et al.\ \cite{les1} underscores the
importance of two additional properties of complex networks above
and beyond more traditionally studied phenomena such as the small
world property. A graph $G\ $with $e_{t}$ edges and $n_{t}$ nodes
satisfies a \emph{densification power law} if there is a constant
$a\in (1,2)$ such that $e_{t}$ is proportional to $n_{t}^{a}.$ In
particular, the average degree grows to infinity with the order of
the network (in contrast to say the preferential attachment model,
which generates graphs with constant average degree). In
\cite{les1}, densification power laws were reported in several
real-world networks such as a physics citation graph and the
internet graph at the level of autonomous systems. Another striking
property found in such networks (and also in on-line social
networks; see \cite{kumar2006sae}) is that distances in the networks
(measured by either diameter or average distance) decreases with
time. The usual models such as preferential attachment or copying
models have logarithmically or sublogarithmically growing diameters
and average distances with time. Various models (such as the Forest
Fire \cite{les1} and Kronecker multiplication \cite{les2} models)
were proposed simulating power law degree distribution,
densification power laws, and decreasing distances.

We present a new model, called \emph{Iterated Local Transitivity}
(ILT), for OSNs and other complex networks which
dynamically simulates many of their properties. The present article is the full
version of the proceedings paper \cite{anppc}. Although modelling
has been done extensively for other complex networks such as the
web graph (see \cite{bonato}), models of OSNs
have only recently been introduced (such as those in
\cite{crandal,kumar2006sae,ln}). The central idea behind the ILT
model is what sociologists call \emph{transitivity}: if $u$ is a
friend of $v$, and $v $ is a friend of $w,$ then $u$ is a friend of
$w$ (see, for example, \cite{ove,scott,white}). In its simplest
form, transitivity gives rise to the notion of \emph{cloning}, where
$u$ is joined to all of the neighbours of $v$. In the ILT model,
given some initial graph as a starting point, nodes are repeatedly
added over time which clone \emph{each} node, so that the new nodes form an
independent set. The ILT model not only incorporates transitivity,
but uses only local knowledge in its evolution, in that a new node
only joins to neighbours of an existing node. Local knowledge is an
important feature of social and complex networks, where nodes have
only limited influence on the network topology. We stress that our
approach is mathematical rather than empirical; indeed, the ILT
model (apart from its potential use by computer and social scientists as a simplified
model for OSNs) should be of theoretical interest in its own right.

Variants of cloning were considered earlier in duplication models
for protein-protein interactions \cite{beb,bhan,chung,pastor}, and in
copying models for the web graph \cite{bj,krrstu}. There are several
differences between the duplication and copying models and the ILT
model. For one, duplication models are difficult to analyze due to
their rich dependence structure. While the ILT model displays a
dependency structure, determinism makes it more amenable to
analysis. The ILT model may be viewed as simplified snapshot of the
duplication model, where \emph{all} nodes are cloned in a given
time-step, rather than duplicating nodes one-by-one over time.
Cloning all nodes at each time-step as in the ILT model leads to
densification and high clustering, along with bad expansion
properties (as we describe in Subsection~\ref{results}).

We finish the introduction with some asymptotic notation. Let $f$ and $g$ be functions whose domain is some
fixed subset of $\mathbb{R}$. We write $f\in O(g)$ if
\begin{equation*}
\lim \sup_{t\rightarrow \infty }\frac{f(t)}{g(t)}
\end{equation*}%
exists and is finite. We will abuse notation and write $f=O(g)$.
We write $f=\Omega (g)$ if $g=O(f)$, and $f=\Theta (g)$ if $f=O(g)$ and $%
f=\Omega (g).$ If $\lim_{t\rightarrow \infty }\left\vert \frac{f(t)}{g(t)}%
\right\vert =0,$ then $f=o(g)$ (or $g=\omega (f)$). So if $f=o(1),$ then $f$
tends to $0.$

\subsection{The ILT Model}

We now give a precise formulation of the model. The ILT model
generates finite, simple, undirected graphs $(G_{t}:t\geq 0).$ \emph{Time-step} $t$, for $t\ge 1$, is defined to be the transition between $G_{t-1}$ and
$G_{t}.$ (Note that a directed graph model
will be considered in the sequel. See also Section~3.)
The only
parameter of the model is the initial graph $G_{0},$ which is any
fixed finite \emph{connected} graph. Assume that for a fixed $t\geq
0,$ the graph $G_{t}$ has been constructed. To form $G_{t+1},$ for
each node $x\in V(G_{t}),$ add its \emph{clone} $x^{\prime },$ such
that $x^{\prime }$ is joined to $x$ and all of its neighbours at
time $t.$ Note that the set of new nodes at time $t+1$ form an
independent set of cardinality $|V(G_t)|.$ See Figure~1 for the graphs
generated from the $4$-cycle over the time-steps $t=1,$ $2,$ $3,$ and $4.$

\begin{figure}[h]
\centering
\begin{tabular}{cc}
\epsfig{file=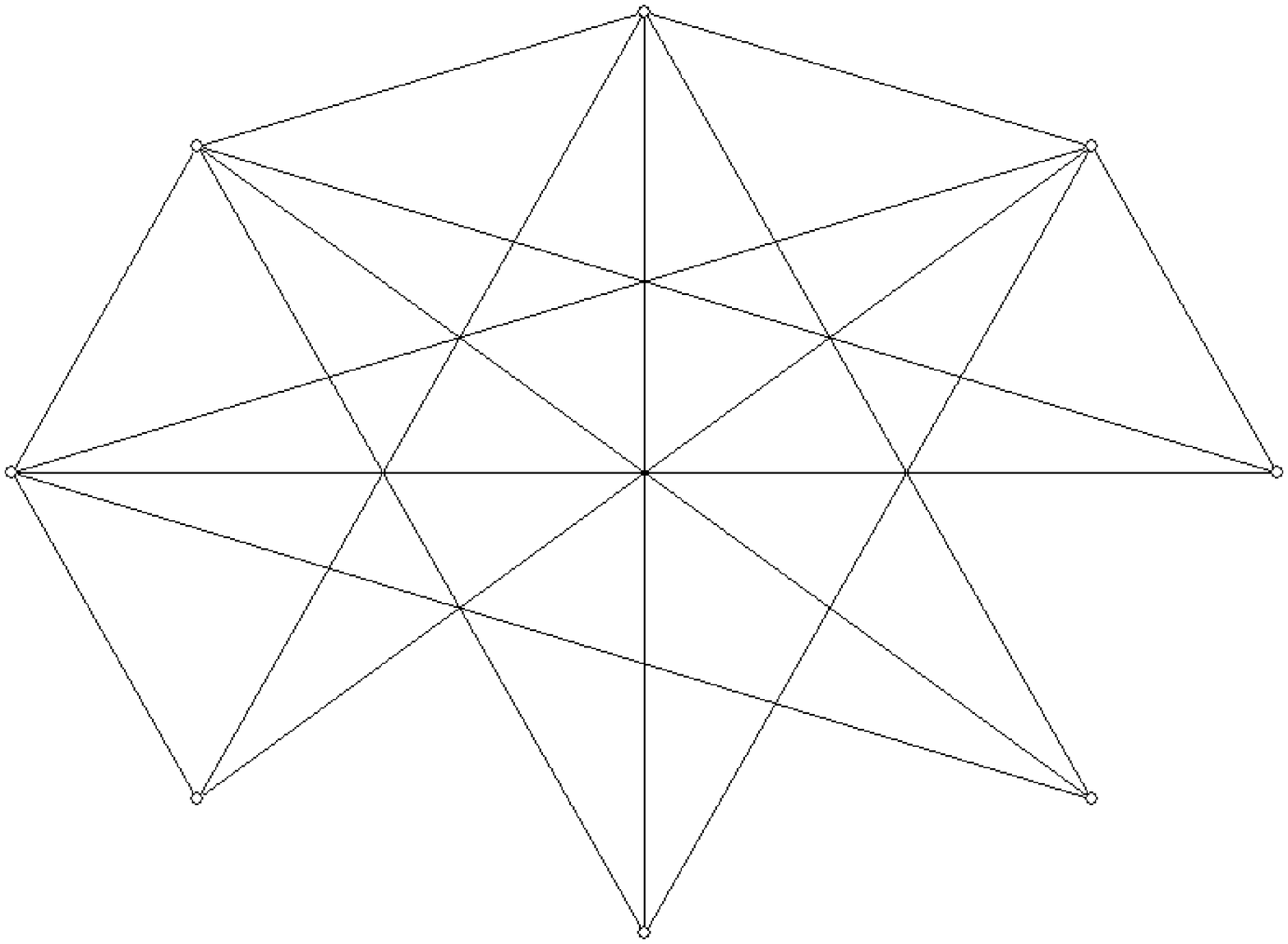,width=2.5in, height=2.5in} &
\epsfig{file=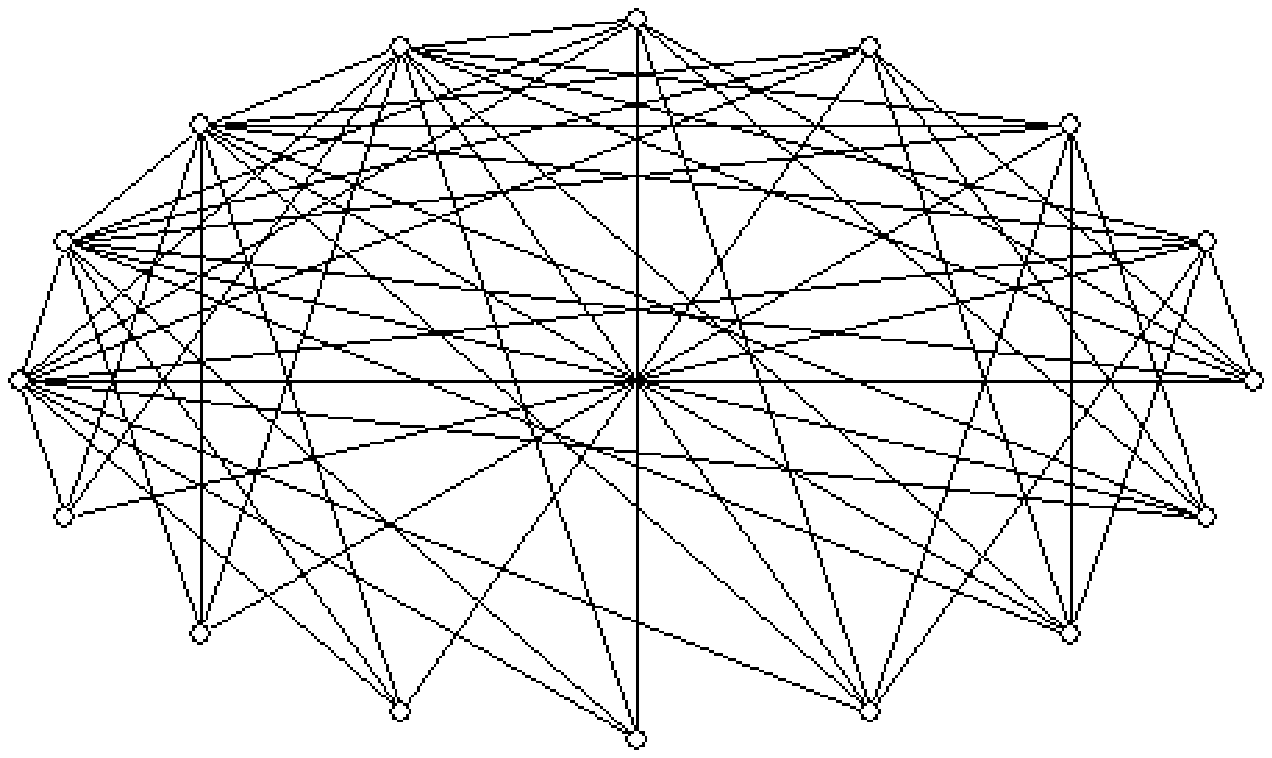,width=2.5in, height=2.5in} \\
\epsfig{file=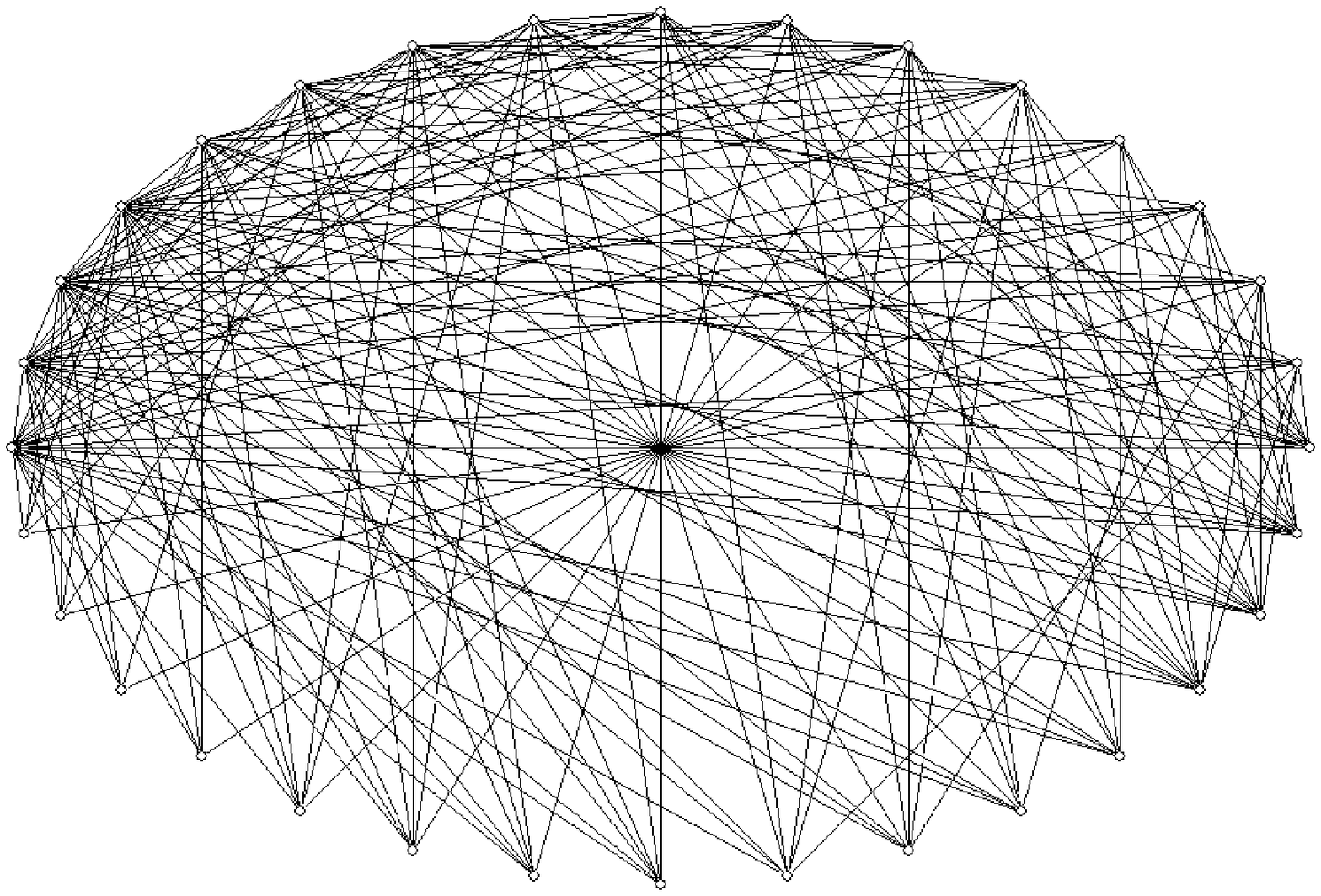,width=2.5in, height=2.5in} &
\epsfig{file=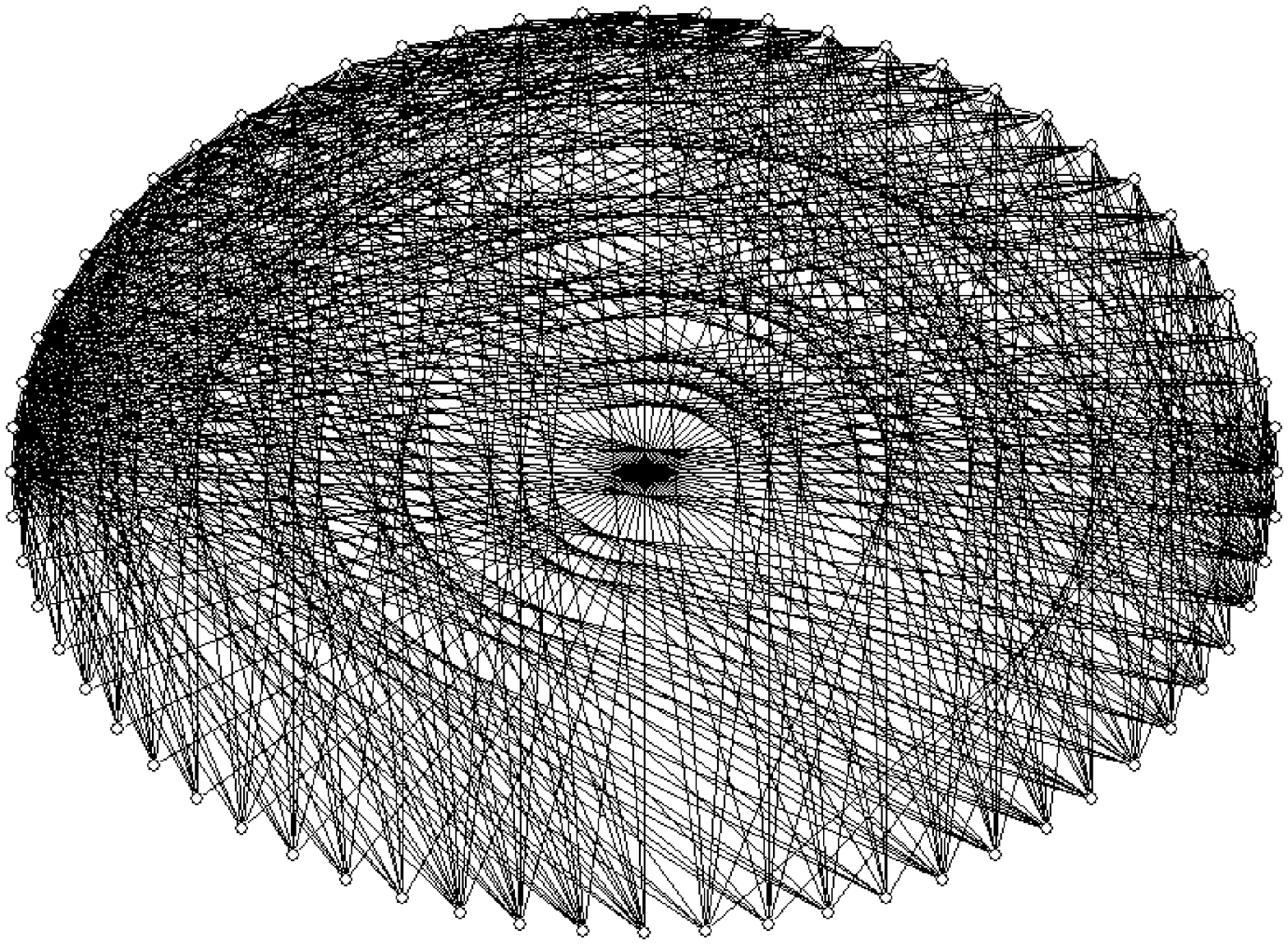,width=2.5in, height=2.5in} \\
\end{tabular}
\caption{The evolution of the ILT model with $G_0=C_4,$ for $t=1,2,3,4.$}
\end{figure}

We write $\deg _{t}(x)$ for the degree of a node at time $t,$
$n_{t}$ for the order of $G_{t},$ and $e_{t}$ for its number of
edges. It is straightforward to see that $n_t=2^tn_0$. Given a node
$x$ at time $t,$ let $x^{\prime }$ be its clone. The elementary but
important recurrences governing the degrees of nodes are given as

\begin{eqnarray}
\deg _{t+1}(x) &=&2\deg _{t}(x)+1,  \label{one} \\
\deg _{t+1}(x^{\prime }) &=&\deg _{t}(x)+1. \label{two}
\end{eqnarray}

\subsection{Main Results}\label{results}

We state our main results on the ILT model,
with proofs deferred to the next section. We give rigorous proofs that the ILT model generates graphs satisfying a densification power law and in many cases decreasing average distance (properties shared by the Forest Fire \cite{les1} and Kronecker multiplication \cite{les2} models). A randomized version of the ILT model is introduced with tuneable densification power law exponent. Properties of the ILT model not shown in the models of \cite{les1,les2} are higher clustering than in random graphs with the same average degree,
and smaller spectral gaps for both their normalized Laplacian
and adjacency matrices than in random graphs. Further, the cop and domination numbers are shown to remain the same
as the graph from the initial graph $G_0$, and the automorphism group
of $G_0$ is a subgroup of the automorphism group of graphs generated at all later times. The ILT model does not, however (unlike the models of \cite{les1,les2}) generate graphs with a power law degree distribution. The number of nodes in the ILT model grows exponentially with time (as in the Kronecker multiplication model, but unlike in the Forest Fire model).

We first demonstrate that the model exhibits a densification power law.  Define the \emph{volume} of $G_{t}$ by
$$
\mathrm{vol}(G_{t})=\sum_{x\in V(G_{t})}\deg _{t}(x) =2e_t.
$$
\begin{theorem}
\label{adeg}For $t>0,$ the average degree of $G_t$ equals
$$
\left( \frac{3}{2}\right) ^{t}\left(
\frac{\mathrm{vol}(G_{0})}{n_{0}} +2\right) -2.
$$
\end{theorem}
Note that Theorem~\ref{adeg} supplies a densification power law
with exponent $a=\frac{\log{3}}{\log{2}} \approx 1.58.$ We think
that the densification power law makes the ILT model realistic,
especially in light of real-world data mined from complex networks
(see \cite{les1}).

We study the average distances and clustering coefficient of the model
as time tends to infinity. Define the \emph{Wiener index} of $G_t$ as
$$W(G_t)=\frac{1}{2}\sum\limits_{x,y\in V(G_t)}d(x,y).$$ The Wiener index may be used
to define the \emph{average distance} of $G_t$ as
$$L(G_t) = \frac{W(G_t)}{ \binom{n_t}{2} }. $$ We will compute the average
distance by deriving first the Wiener index. Define the
\emph{ultimate average distance of} $G_{0}$, as
$$
UL(G_{0})=\lim_{t\rightarrow \infty }L(G_{t})
$$
assuming the limit exists. Note that the ultimate average distance is a new graph
parameter. We provide an exact value for $L(G_{t})$
and compute the ultimate average distance for any initial graph
$G_0.$

\begin{theorem}\label{adist}
$ $\\
\begin{enumerate}
\item For $t>0,$
\begin{eqnarray*}
W(G_{t}) &=&4^{t}\left( W(G_{0})+(e_{0}+n_{0})\left( 1-\left(
\frac{3}{4} \right) ^{t}\right) \right) .
\end{eqnarray*}

\item For $t>0,$
$$
L(G_{t})= \frac{4^{t}\left( W(G_{0})+(e_{0}+n_{0})\left(
1-\left( \frac{3}{4}\right) ^{t}\right) \right)
}{4^{t}n_{0}^{2}-2^{t}n_{0}}.
$$

\item For all graphs $G_{0},$
$$
UL(G_{0})=\frac{W(G_{0})+e_{0}+n_{0}}{n_{0}^{2}}.
$$
Further, $UL(G_{0})\leq L(G_{0})$ if and only if $W(G_{0}) \ge
(n_0-1)(e_0+n_0).$
\end{enumerate}
\end{theorem}

Note that the average distance of $G_{t}$ is bounded above by
$\mathrm{diam} (G_{0})+1$ (in fact, by $\mathrm{diam}(G_{0})$ in all
cases except cliques). Further, the condition in (3) for $UL(G_{0})<
L(G_{0})$ holds for large cycles and paths. Hence, for many initial
graphs $G_{0},$ the average distance decreases, a property observed
in OSNs and other complex networks (see \cite{kumar2006sae,les1}).

Let $N_t(x)$ be the neighbour set of $x$ at
time $t$, let $G_{t}\upharpoonright N_t(x)$ be the subgraph induced by
$N_t(x)$ in $G_t,$ and let $e(x,t)$ be the number of edges in
$G_{t}\upharpoonright N_t(x).$ For a node $x\in V(G_{t})$ with degree
at least $2$ define
$$
c_t(x)=\frac{e(x,t)}{{{\deg _{t}(x)}\choose{2}} }.
$$
By convention $c_t(x)=0$ if the degree of $x$ is at most $1.$ The
\emph{clustering coefficient} of $G_t$ is $$C(G_t) = \frac{ \sum_{x
\in V(G_t)}c_t(x)}{n_t}.$$
The clustering coefficient of the graph at time $t$ generated by the ILT model is estimated and
shown to tend to $0$ slower than a $G(n,p)$ random graph with the
same average degree.

\begin{theorem}\label{cluster}
$$
\Omega \left( \left( \frac{7}{8}\right) ^{t}t^{-2}\right)
=C(G_{t})=O\left( \left( \frac{7}{8}\right) ^{t}t^{2}\right) .
$$
\end{theorem}
\noindent Observe that $C(G_t)$ tends to $0$ as $t\rightarrow \infty .$ If we
let $ n_{t}=n$ (so $t\sim \log_2 n ),$ then this gives that
\begin{equation*}
C(G_{t})= n^{\log _{2}(7/8) + o(1)}.
\end{equation*}
In contrast, for a random graph $G(n,p)$ with comparable average
degree $$pn=\Theta ((3/2)^{\log _{2}n})=\Theta (n^{\log _{2}(3/2)})$$
as $G_{t}$, the clustering coefficient is $p=\Theta (n^{\log
_{2}(3/4)})$ which tends to zero much faster than $C(G_{t}).$ (For a discussion of the
clustering coefficient of $G(n,p)$, see Chapter~2 of \cite{bonato}.)

Social networks often organize into separate clusters in which the
intra-cluster links are significantly higher than the number of
inter-cluster links. In particular, social networks contain
communities (characteristic of social organization), where tightly
knit groups correspond to the clusters \cite{gn}. As a result,
social networks possess bad expansion properties realized by small
gaps between their first and second eigenvalues \cite{estrada}. We find that the ILT model has bad expansion properties as indicated by the spectral gap of both its normalized Laplacian and adjacency matrices.

For regular graphs, the eigenvalues of the adjacency matrix are
related to several important graph properties, such as in the
expander mixing lemma. The normalized Laplacian of a graph,
introduced by Chung~\cite{sgt}, relates to important graph
properties even in the case where the underlying graph is not
regular (as is the case in the ILT model).  Let $A$ denote the
adjacency matrix and $D$ denote the diagonal adjacency matrix of a
graph $G$.  Then the normalized Laplacian of $G$ is
\[
\mathcal{L} = I - D^{-1/2}AD^{-1/2}.
\]
Let $0 = \lambda_0 \leq \lambda_1 \cdots \leq \lambda_{n-1} \leq 2$ denote
the eigenvalues of $\mathcal{L}$.  The \emph{spectral gap} of the
normalized Laplacian is $$\lambda = \max\{ |\lambda_1 - 1|,
|\lambda_{n-1} - 1| \}.$$  Chung, Lu, and Vu \cite{clv} observe
that, for random power law graphs with some parameters (effectively
in the case that $d_{\min} = c\log^{2}n$ for some constant $c>0$ and all integers $n>0$), that $\lambda \leq (1 +
o(1)) \frac{4}{\sqrt{d}},$ where $d$ is the average degree.

For the graphs $G_{t}$ generated by the ILT model, we observe that the spectra
behaves quite differently and, in fact, the spectral gap has a
constant order. The following theorem suggests a significant
spectral difference between graphs generated by the ILT model and
random graphs. Define $\lambda(G_t)$ to be the spectral gap of the
normalized Laplacian of $G_t.$
\begin{theorem}\label{expa}  For $t \geq 1$, $\lambda(G_t) > \frac{1}{2}$.
\end{theorem}

Theorem~\ref{expa} represents a drastic departure from the good
expansion found in random graphs, where $\lambda = o(1)$
\cite{sgt,clv,fk}, and from the preferential attachment model
\cite{mihail}. If $G_0$ has bad expansion properties, and has $\lambda_{1} < 1/2$
(and thus, $\lambda
> 1/2$) then, in fact, this trend of bad expansion continues as shown by the following
theorem.
\begin{theorem} \label{thm:decgap}  Suppose $G_0$ has at least two nodes, and for $t>0$
let $\lambda_{1}(t)$ be the second eigenvalue of $G_{t}.$  Then we
have that
\[
\lambda_1(t) < \lambda_1(0).
\]
\end{theorem}

Note that Theorem~\ref{thm:decgap} implies that $\lambda_1(1) <
\lambda_1(0)$ and this implies that the sequence $\{\lambda_1(t):
t\geq 0\}$ is strictly decreasing.  This follows since
$G_{t}$ is constructed from $G_{t-1}$ in the same manner as $G_{1}$
is constructed from $G_0$.   If $G_0$ is $K_1$, then there is no
second eigenvalue, but $G_1$ is $K_2.$ Hence, in this case, the
theorem implies that $\{\lambda_1(t): t \geq 1\}$ is strictly
decreasing.

Let $\rho_0(t) \geq |\rho_1(t)| \geq \dots$ denote the eigenvalues
of the adjacency matrix $G_t.$   If
$A$ is the adjacency matrix of $G_t,$ then the adjacency matrix of
$G_{t+1}$ is
\[
M = \left(
\begin{array}{cc}
A & A+I \\
A+I & 0%
\end{array}
\right) ,
\]
where $I$ is the identity matrix of order $n_t$. We note the
following recurrence for the eigenvalues of the adjacency matrix of
$G_{t}$.
\begin{theorem} \label{adj1}
If $\rho$ is an eigenvalue of the adjacency matrix of $G_t$, then
\[
\frac{\rho \pm \sqrt{\rho^{2} + 4(\rho+1)^{2}}}{2},
\]
are eigenvalues of the adjacency matrix of $G_{t+1}$.
\end{theorem}
\noindent We leave the reader to check that the eigenvectors of $G_{t}$ can be
written in terms of the eigenvectors of $G_{t-1}$. As in the Laplacian case, we show that there is a small spectral gap of the adjacency matrix.
\begin{theorem}\label{adj} Let $\rho_0(t) \geq |\rho_1(t)| \geq \dots \geq
|\rho_{n}(t)|$ denote the eigenvalues of the adjacency matrix of
$G_t$.  Then
\[\frac{\rho_0(t)}{|\rho_1(t)|} = \Theta(1). \]
\end{theorem}
\noindent That is, $\rho_1(t) \geq c|\rho_0(t)|$ for some
constant $c>0.$ Theorem~\ref{adj} is in contrast to the fact that in
$G(n,p)$ random graphs, $|\rho_1| = o(\rho_0)$ (see \cite{sgt}).

In a graph $G$, a set $S$ of nodes is a \emph{dominating} set if every node not
in $S$ has a neighbour in $S$. The \emph{domination number} of $G$, written $\gamma (G)$, is
the minimum cardinality of a
dominating set in $G$. We use $S$ to represent a dominating set in $G$, where each node not in $S$ is joined
to some node of $S$. A graph parameter bounded below by the domination number is the
so-called cop (or search) number of a graph. In Cops and Robbers, there are two players, a set of $s$ \emph{cops} (or \emph{searchers}) $\mathcal{C}$, where $s>0$ is a fixed integer, and the \emph{robber} $\mathcal{R}.$ The cops begin the game by occupying a set of $s$ nodes of a simple, undirected, and finite  graph $G$. While the game may be played on a disconnected graph, without loss of generality, assume that $G$ is connected (since the game is played independently on each component and the number of cops required is the sum over all components). The cops and robber move in \emph{rounds} indexed by non-negative integers. Each round consists of a cop's move followed by a robber's move. More than one cop is allowed to occupy a node, and the players may \emph{pass}; that is, remain on their current nodes. A \emph{move} in a given round for a cop or the robber consists of a pass or moving to an adjacent node; each cop may move or pass in a round. The players know each others current locations; that is, the game is played with \emph{perfect information}. The cops win and the game ends if at least one of the cops can eventually occupy the same node as the robber; otherwise, $\mathcal{R}$ wins. As placing a cop on each node guarantees that the cops win, we may define the \emph{cop number}, written $c(G),$ which is the minimum cardinality of the set of cops needed to win on $G.$ While this node pursuit game played with one cop was introduced in~\cite{nw,q}, the cop number was first introduced in~\cite{af}. For a survey of results on Cops and Robbers, see
\cite{gena}.

We prove that the domination and cop numbers of $G_t$ depend only on the initial graph $G_0$. Theorem~\ref{dom} shows that even as the graph becomes large as $t$ progresses, the same number of nodes needed at time $0$ to dominate the graph will be needed at time $t$.
\begin{theorem} \label{dom}
For all $t \ge 0$,
$$ \gamma(G_t) =\gamma (G_0),$$
and
$$c(G_t)=c(G_0).$$
\end{theorem}
In Theorem~\ref{dom}, we prove that the cop number remains the same for $G_t$. This implies that no matter how large the graph $G_t$ becomes, the robber can be captured by the same number of cops used at time $0$. In terms of OSNs, Theorem~\ref{dom} suggests that users in the network can easily spread and track information (such as gossip) no matter how large the graph becomes.

An \emph{automorphism} of a graph $G$ is an isomorphism from $G$ to itself; the set of all automorphisms forms a group under the operation of composition,
written $ \mathrm{Aut(G)}$. We say that an automorphism $f_t \in \mathrm{Aut}(G_t)$ \emph{extends} to $f_{t+1} \in \mathrm{Aut}(G_{t+1})$ if $$f_{t+1} \upharpoonright V(G_t) =f_t;$$
that is, the restriction of the map $f_{t+1}$ to $V(G_t)$ equals $f_t.$ We show that symmetries from $t=0$ are preserved at time $t$. This provides further evidence that the ILT model retains a memory of the initial graph from time $0.$

\begin{theorem} \label{embed}
For all $t \ge 0 $, $\mathrm{Aut}(G_0)$ embeds in $\mathrm{Aut}(G_t)$.
\end{theorem}

As shown in Theorem~\ref{adeg}, the ILT model has a fixed densification exponent equalling $\log 3 / \log 2$. We consider a randomized version of the model which allows for this exponent to become tuneable. To motivate the model, in OSNs some new users are friends outside of the OSN. Such users immediately seek each other out as they join the OSN and become friends there. The stochastic model ILT($p$) is defined as follows. Define $H_0$ to be $K_1.$ A sequence $(H_t: t\in \mathbb{N})$ of graphs is generated so that for all $t$, $H_t$ is an induced subgraph of $H_{t+1}.$ At time $t+1,$ first clone all the nodes of $H_t$ as in the deterministic ILT model. Let $n$ be the number new nodes are added at time $t+1.$ (Note that $n$ is a function of $t$ and is not a new parameter.) To form $H_{t+1}$, add edges independently between the new nodes with probability $p=p(n).$ Hence, the new nodes form a random graph $G(n,p)$.

Several properties of the ILT model are inherited by the ILT($p$) model. For example, as we are adding edges to the graphs generated by the ILT model, the average distance may only decrease, and the clustering coefficient may only increase. The following theorem proves that ILT($p$) generates graphs following a densification power law with exponent $\log (3+\delta)/\log 2, $ where $0 \le \delta \le 1$.  For $T$ a positive integer representing time, we say that an event holds \emph{asymptotically
almost surely} (\emph{a.a.s.}) if the probability that it holds
tends to $1$ as $T$ tends to infinity.
\begin{theorem}\label{random_dens}
Let $0 \le \delta \le 1$, and define
\begin{equation}\label{eq:p}
p(n) = \delta n^{\frac{\log (3+\delta)}{\log 2}} / n^2.
\end{equation} Then a.a.s.
\begin{equation*}
\vol (H_T) = (1+o(1))(3+\delta)^{T}.
\end{equation*}
\end{theorem}

Hence, by choosing an appropriate $p$, the densification power law exponent in graphs generated by the ILT($p$) model may achieve any value in the interval $[\log 3/ \log 2, 2]$. We also prove that for the normalized Laplacian, the ILT($p$) model maintains a large spectral gap.

\begin{theorem}\label{random_gap}
A.a.s.
$$\lambda (H_T)= \Omega(1).$$
\end{theorem}

\section{Proofs of Results}

This section is devoted to the proofs of the theorems outlined in Section~1.

\subsection{Proof of Theorem~\ref{adeg}}

We now consider the number of edges and average degree of $G_{t},$
and prove the following densification power law for the ILT model.
Define the \emph{volume} of $G_{t}$ by
$$
\mathrm{vol}(G_{t})=\sum_{x\in V(G_{t})}\deg _{t}(x) =2e_t.
$$

The proof of Theorem~\ref{adeg} follows directly from the following
Lemma~\ref{lemm}, since the average degree of $G_t$ is
$\mathrm{vol}(G_t)/n_t.$

\begin{lemma}
\label{lemm}For $t>0,$
$$
\mathrm{vol}(G_{t})=3^{t}\mathrm{vol}(G_{0})+2n_{0}(3^{t}-2^{t}).
$$
In particular,
\begin{eqnarray*}
e_{t} &=&3^{t}(e_{0}+n_{0})-n_{t}.
\end{eqnarray*}
\end{lemma}

\begin{proof}
By (\ref{one}) and (\ref{two}) we have that
\begin{eqnarray}
\mathrm{vol}(G_{t+1}) &=&\sum_{x\in V(G_{t})}\deg
_{t+1}(x)+\sum_{x^{\prime }\in
V(G_{t+1})\backslash V(G_{t})}\deg _{t+1}(x^{\prime }) \nonumber \\
&=&\sum_{x\in V(G_{t})}(2\deg _{t}(x)+1)+\sum_{x\in V(G_{t})}(\deg
_{t}(x)+1) \nonumber
\\
&=&3\mathrm{vol}(G_{t})+n_{t+1}. \label{recc}
\end{eqnarray}
Hence by (\ref{recc}) for $t>0,$
\begin{eqnarray*}
\mathrm{vol}(G_{t}) &=&3\mathrm{vol}(G_{t-1})+n_{t} \\
&=&3^{t}\mathrm{vol}(G_{0})+n_{0}\left( \sum_{i=0}^{t-1}3^{i}2^{t-i}\right) \\
&=&3^{t}\mathrm{vol}(G_{0})+2n_{0}(3^{t}-2^{t}),
\end{eqnarray*}
where the third equality follows by summing a geometric series.
\end{proof}

\subsection{Proof of Theorem~\ref{adist}}

When computing distances in the ILT model, the following lemma is
helpful.

\begin{lemma}\label{k}
Let $x$ and $y$ be nodes in $G_{t}\ $with $t>0.$ Then
$$
d_{t+1}(x^{\prime },y)=d_{t+1}(x,y^{\prime
})=d_{t+1}(x,y)=d_{t}(x,y),
$$
and
$$
d_{t+1}(x^{\prime },y^{\prime })=\left\{
\begin{array}{cc}
d_{t}(x,y) & \text{if }xy\notin E(G_{t}) ,\\
d_{t}(x,y)+1=2 & \text{if }xy\in E(G_{t}).
\end{array}
\right.
$$
\end{lemma}

\begin{proof} We prove that $d_{t+1}(x,y)=d_{t}(x,y)$. The proofs of the other equalities are analogous and so omitted. Since in the ILT model we do not delete any edges, the distance cannot increase after a ``cloning" step occurs. Hence,
$d_{t+1}(x,y) \le d_{t}(x,y).$ Now suppose for a contradiction that there is a path $P'$ connecting $x$ and $y$ in $G_{t+1}$ with length $k < d_{t}(x,y).$
Hence, $P'$ contains nodes not in $G_t$. Choose such a $P'$ with the least number of nodes, say $s>0$, not in $G_t$. Let $z'$ be a node of $P'$ not in $G_t$, and let the neighbours
of $z'$ in $P'$ be $u$ and $v.$ Then $z \in V(G_t)$ is joined to $u$ and $v.$ Form the path $Q'$ by replacing $z'$ by $z$. But then $Q'$ has length $k$ and has $s-1$ many nodes not in $G_t$, which supplies a contradiction. \end{proof}

We now turn to the proof of Theorem~\ref{adist}. We only prove item
(1), noting that items (2) and (3) follow from (1) by computation.
We derive a recurrence for $W(G_{t})$ as follows. To compute
$W(G_{t+1}),$ there are five cases to consider: distances within $
G_{t},$ and distances of the forms: $d_{t+1}(x,y^{\prime }),$
$d_{t+1}(x^{\prime },y),$ $ d_{t+1}(x,x^{\prime }),$ and
$d_{t+1}(x^{\prime },y^{\prime }).$ The first three cases contribute
$3W(G_{t})$ by Lemma~\ref{k}. The 4th case contributes $n_{t}.$ The
final case contributes $W(G_{t})+e_{t}$ (the term $e_{t}$ comes from
the fact that each edge $xy$ contributes $d_{t}(x,y)+1).$

Thus,
\begin{eqnarray*}
W(G_{t+1}) &=&4W(G_{t})+e_{t}+n_{t} \\
&=&4W(G_{t})+3^{t}(e_{0}+n_{0}).
\end{eqnarray*}

Hence,
\begin{eqnarray*}
W(G_{t}) &=&4^{t}W(G_{0})+ \sum_{i=0}^{t-1}4^{i}\left(
3^{t-1-i}\right) (e_{0}+n_{0})  \\
&=&4^{t}W(G_{0})+4^{t}(e_{0}+n_{0})\left( 1-\left(
\frac{3}{4}\right) ^{t}\right). \qed \qedhere
\end{eqnarray*}

Diameters are constant in the ILT model. We record this as a strong
indication of the (ultra) small world property in the model.

\begin{lemma}
For all graphs $G_0$ different than a clique,
$$
\mathrm{diam}(G_{t})=\mathrm{diam}(G_{0}),
$$
and $\mathrm{diam}(G_{t})=\mathrm{diam}(G_{0})+1=2$ when $G_{0}$ is
a clique.
\end{lemma}

\begin{proof} This follows directly from Lemma~\ref{k}. \end{proof}

\subsection{Proof of Theorem~\ref{cluster}}

We introduce the following dependency structure that will help us
classify the degrees of nodes. Given a node $x \in V(G_0)$ we define
its \emph{descendant tree at time} $t$, written $T(x,t)$, to be a
rooted binary tree with root $x$, and whose leaves are all of the
nodes at time $t$. To define the $(k+1)$th row of $T(x,t)$, let $y$
be a node in the $k$th row ($y$ corresponds to a node in $G_k$).
Then $y$ has exactly two descendants on row $k+1$: $y$ itself and
$y^{\prime }.$ In this way, we may identify the nodes of $G_{t}$
with a length $t$ binary sequence corresponding to the descendants
of $x,$ using the convention that a clone is labelled $1.$ We refer
to such a sequence as the \emph{binary sequence for} $x$ at time
$t.$ We need the following technical lemma.

\begin{lemma}\label{degree}
Let $S(x,k,t)$ be the nodes of $T(x,t)$ with exactly $k$ many $0$'s
in their binary sequence at time $t.$ Then for all $y\in S(x,k,t)$
$$
2^k (\deg_0(x)+1)+t-k-1 ~\leq~ \deg _{t}(y) ~\leq~ 2^k
(\deg_0(x)+t-k+1)-1.
$$
\end{lemma}

\begin{proof}
The degree $\deg _{t}(y)$ is minimized when $y$ is identified with the binary sequence beginning with $k$ many $0$'s: $(0,\ldots ,0,1,1,\ldots ,1).$ In this case,
\begin{eqnarray*}
\deg _{t}(y) &=&2(2(\dots(2(2 \deg_0(x)+1)+1)\dots)+1)+1 + (t-k)\\
&=&2^k (\deg_0(x)+1)+t-k-1.
\end{eqnarray*} The degree $\deg _{t}(y)$ is maximized when the
sequence with the $k$ many $0$'s at the end of the sequence: $(1,1,\ldots ,1,0,\ldots ,0).$ Then
\begin{eqnarray*}
\deg _{t}(y) &=&2(2(\dots(2(\deg_0(x)+t-k)+1)\dots)+1)+1\\
&=&2^k (\deg_0(x)+t-k+1)-1. \qedhere
\end{eqnarray*}
\end{proof}

It can be shown (using Lemma~\ref{degree}) that the number of nodes of
degree at least $j$ at time $t,$ denoted by $N_{(\ge j)},$ satisfies
$$
\sum_{i=\log_2 j}^{t} {t \choose i} \le N_{(\ge j)} \le
\sum_{i=\max\{\log_2 j-\log_2 t-O(1),0\}}^{t} {t \choose i}.
$$
Indeed, when a vertex is identified with the binary sequence with $i \ge
\log_2 k$ many $0$'s, then the degree is at least $k$. We have $t \choose i$ such sequences. On the other hand, if the binary sequence has $i \le \log_2 k - \log_2 t - O(1)$ many $0$'s, then the corresponding vertex has degree smaller than $k$.
In particular, $N_{(\ge j)} = \Theta(n_t)$ for $j \le \sqrt{n_t},$
and therefore, the degree distribution of $G_t$ does not follow a
power law. Since ${t \choose j}$ nodes have degree around $2^j,$ the
degree distribution has ``binomial-type'' behaviour. As an example of the degree distribution of a graph generated by the ILT model, see Figure~2.

\begin{figure}[h]
\centering
\epsfig{figure=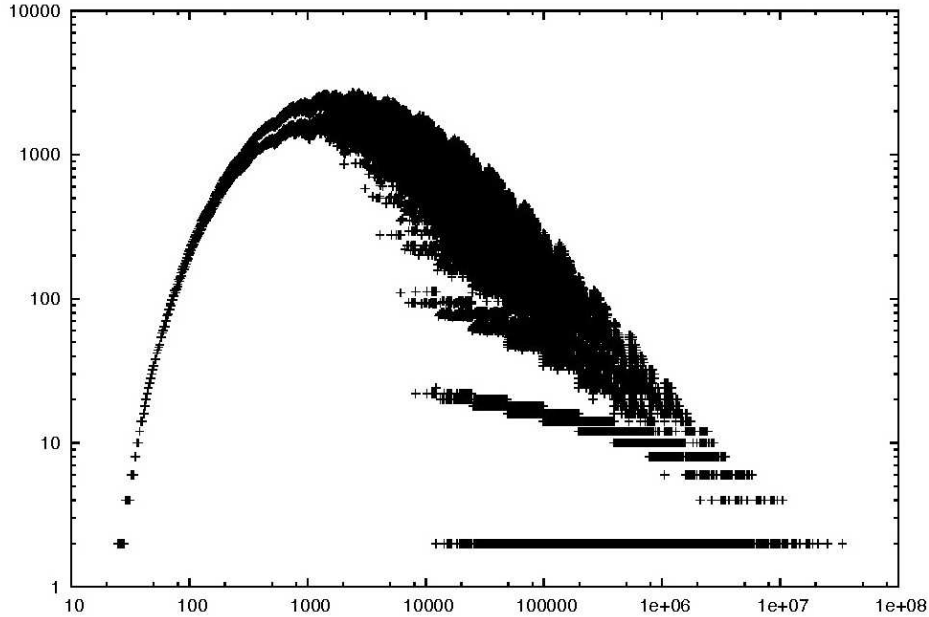}
\caption{A log-log plot of the degree distribution for $G_{25}$ with $G_0 = K_1.$}
\end{figure}

We now prove the
following lemma. Recall that $e(x,t)$ is the number of edges in
$G_{t}\upharpoonright N_t(x).$
\begin{lemma}\label{uu}
For all $x\in V(G_{t})$ with $k$ $0$'s in their binary sequence, we
have that
$$
\Omega (3^{k})=e(x,t)=O(3^{k}t^{2}).
$$
\end{lemma}

We note that the constants hidden in $\Omega(\cdot)$ and $O(\cdot)$ notations (both in the statement of the lemma and in the proof below) do not depend on $k$ nor $t.$
\bigskip

\noindent \emph{Proof of Lemma~\ref{uu}.}
For $x\in V(G_{t})$ we have that
\begin{eqnarray*}
e(x,t+1) &=&e(x,t)+\deg _{t}(x)+\sum_{i=1}^{\deg _{t}(x)}(1+\deg
_{G_{t}\upharpoonright N_t(x)}(x)) \\
&=&3e(x,t)+2\deg _{t}(x).
\end{eqnarray*}
For $x^{\prime },$ we have that
$$
e(x^{\prime },t+1)=e(x,t)+\deg _{t}(x).
$$
Since there are $k$ many $0$'s and $e(x,2)$ is always positive for
all initial graphs $G_0$, $e(x,t) \ge 3^{k-2}e(x,2) = \Omega(3^k)$
and the lower bound follows.

For the upper bound, a general binary sequence corresponding to $x$
is of the form
$$
(1,\ldots ,1,0,1,\ldots ,1,0,1,...,1,0,1,\ldots ,1,0,1,\ldots ,1)
$$
with the $0$'s in positions $i_{k}$ ($1\leq i\leq k$). Consider a
path in the descendant tree from the root of the tree to node $x$.
By Lemma~\ref{degree},  the node on the path in the $i$th row ($i <
i_{j}$) has (at time $i$) degree $O(2^{j-1}t)$.

Hence, the number of edges we estimate is $O(t^2)$ until the
$(i_1-1)$th row, increases to $3 O(t^2) + O(2^1 t)$ in the next row,
and increases to $3 O(t^2) + O(2^1 t^2)$ in the $(i_2-1)$th row. By
induction, we have that
\begin{eqnarray*}
e(x,t) &=& 3(\dots(3(3 O(t^2)+O(2^1 t^2))+O(2^2 t^2)) \dots)+O(2^k t^2) \\
&=&O(t^{2})3^{k}\sum_{i=0}^{k}\left( \frac{2}{3}\right) ^{j} \\
&=&O(3^{k}t^{2}).\qedhere \qed
\end{eqnarray*}

We now prove our result on clustering coefficients.

\noindent \emph{Proof of Theorem~\ref{cluster}.} For $x\in V(G_{t})$
with $k$ many $0$'s in its binary sequence, by Lemmas~\ref{degree}
and \ref{uu} we have that
$$
c_t(x)=\Omega \left( \frac{3^{k}}{\left( 2^{k}t\right) ^{2}}\right)
=\Omega \left( \left( \frac{3}{4}\right) ^{k}t^{-2}\right) ,
$$
and
$$
c_t(x) = O\left( \frac{3^{k}t^{2}}{\left( 2^{k}\right) ^{2}}\right)
=O\left( \left( \frac{3}{4}\right) ^{k}t^{2}\right) .
$$

Hence, since we have $n_0 {t \choose k}$ nodes with $k$ many $0$'s
in its binary sequence,
$$
C(G_{t})=\frac{\sum_{k=0}^{t} n_0 {{t}\choose{k}} \Omega \left(
\left( \frac{3}{4} \right) ^{k}t^{-2}\right) }{n_0 2^{t}}=\Omega
\left( \frac{t^{-2}\left( 1+\frac{3 }{4}\right) ^{t}}{2^{t}}\right)
=\Omega \left( \left( \frac{7}{8}\right) ^{t}t^{-2}\right) .
$$
In a similar fashion, it follows that
$$
C(G_{t})=\frac{\sum_{k=0}^{t} n_0 {{t}\choose{k}} O\left( \left(
\frac{3}{4}\right) ^{k}t^{2}\right) }{n_0 2^{t}}=O\left( \left(
\frac{7}{8}\right) ^{t}t^{2}\right) . \qed \qedhere
$$

\subsection{Proofs of Theorems~\ref{expa}, \ref{thm:decgap}, \ref{adj1}, and \ref{adj}}\label{spec}

We present proofs of the spectral properties of the ILT model. For ease of notation, let $\lambda(t) = \lambda(G_t).$ \vspace{.1in}

\noindent \textbf{Proof of Theorem~\ref{expa}. }
We use the expander mixing lemma for the normalized
Laplacian (see \cite{sgt}). For sets of nodes $X$ and $Y$ we use the notation
$\vol(X)$ for the volume of the subgraph induced by $X,$ and $e(X,Y)$ for the number of edges with one end in
each of $X$ and $Y.$
\begin{lemma}\label{mix}
For all sets $X \subseteq G,$
\[
\left| e(X,X) - \frac{(\vol(X))^{2}}{\vol(G)} \right| \leq \lambda
\frac{\vol(X)\vol(\bar{X})}{\vol(G)}.
\]
\end{lemma}

We observe that
$G_{t}$ contains an independent set (that is, a set of nodes with no edges) with volume $\vol(G_{t-1}) +
n_{t-1}$. Let $X$ denote this set, that is, the new nodes added at
time $t$. Then by (\ref{recc}) it follows that $$\vol(\bar{X}) = \vol(G_t)-\vol(X)=2\vol(G_{t-1}) +
n_{t-1}.$$  Since $X$ is independent, Lemma~\ref{mix} implies that
\[
\lambda(t) \geq \frac{\vol(X)}{\vol(\bar{X})} = \frac{\vol(G_{t-1})
+ n_{t-1}}{2\vol(G_{t-1}) + n_{t-1}} > \frac{1}{2}. \qed \qedhere
\]

\noindent \textbf{Proof of Theorem~\ref{thm:decgap}.} Before we proceed with the proof of Theorem $\ref{thm:decgap}$, we
begin by stating some notation and a lemma. For a given node $u \in
V(G_{t})$, we let $\tilde{u} \in V(G_{0})$ denote the node in $G_0$
that $u$ is a descendant of.  Given $uv \in E(G_0)$, we define
\[
\mathcal{A}_{uv}(t) = \{ xy \in E(G_{t}): \tilde{x} = u, \tilde{y} =
v\},
\]
and for $v \in E(G_0)$, we set
\[
\mathcal{A}_v(t) = \{ xy \in E(G_{t}): \tilde{x}=\tilde{y} = v\}.
\]

We use the following lemma, for which the proof of items (1) and (2)
follow from Lemma~\ref{lemm}. The final item contains a standard
form of the Raleigh quotient characterization of the second
eigenvalue; see \cite{sgt}.
\begin{lemma}\label{lemmm}
$ $ \\
\begin{enumerate}
\item For $uv \in E(G_0)$,
\[
|\mathcal{A}_{uv}(t)| = 3^{t}.
\]
\item For $v \in V(G_0),$
\[
|\mathcal{A}_{v}(t)| = 3^{t} - 2^{t}.
\]
\item
Define
\begin{eqnarray*}
\bar{d} &=& \frac{ \sum_{v\in V(G_t)} f(v)\deg_t(v)}{\vol(G_t)}.
\end{eqnarray*}
Then
\begin{equation}
\lambda_1(t) = \inf_{\substack{f: V(G_t) \to \mathbb{R},\\ f \neq
0}} \frac{\sum\limits_{uv\in E(G_t)} (f(u) -
f(v))^{2}}{\sum\limits_{v} f^{2}(v) \deg_t(v) - \bar{d}^{2}
\vol(G_t)}. \label{eqn:ral3}
\end{equation}
\end{enumerate}
\end{lemma}
Note that in item (3), $\bar{d}$ is a function of $f.$
Now let $g:V(G_0)
\to \mathbb{R}$ be the harmonic eigenvector for $\lambda_1(0)$ so
that $$\sum_{v\in V(G_0)} g(v) \deg_0(v) = 0,$$ and
\[
\lambda_1(0) = \frac{\sum\limits_{uv \in E(G_0)} (g(u) -
g(v))^{2}}{\sum\limits_{v \in V(G_0)} g^2(v)\deg_0(v)}.
\]
Furthermore, we choose $g$ scaled so that $\sum_{v \in V(G_0)}
g^{2}(v) \deg_0(v) = 1$.  This is the standard version of the
Raleigh quotient for the normalized Laplacian from \cite{sgt}, so
such a $g$ exists so long as $G_0$ has at least two eigenvalues,
which it does by our assumption that $G_0 \ncong K_1$. Our strategy
in proving the theorem is to show that lifting $g$ to
$G_1$ provides an effective bound on the second eigenvalue of
$G_{1}$ using the form of the Raleigh quotient given in
(\ref{eqn:ral3}).

Define $f:G_{t} \to \mathbb{R}$ by $ f(x) = g(\tilde{x}).$ Then note
that
\begin{eqnarray*}
\sum_{xy \in E(G_t)} (f(x) - f(y))^{2} &=& \sum_{\substack{xy \in
E(G_t), \\ \tilde{x} = \tilde{y}}} (f(x) - f(y))^{2} +
\sum_{\substack{xy \in E(G_t)\\ \tilde{x} \neq \tilde{y}}} (f(x) - f(y))^{2} \\
&=& \sum_{uv \in E(G_0)} \sum_{ xy \in \mathcal{A}_{uv }} (g(u) - g(v))^{2} \\
&=& 3^{t}\sum_{uv \in E(G_0)} (g(u)-g(v))^2.
\end{eqnarray*}
By Lemma~\ref{lemmm} (1) and (2) it follows that
\begin{eqnarray*}
\sum_{x \in V(G_t)} f^{2}(x) \deg_t(x) &=&
\sum_{x \in V(G_{t})} \sum_{xy \in E(G_t)} f^{2}(x) \\ &=& \sum_{u \in V(G_{0})} \sum_{\substack{ x y \in E(G_t), \\ \tilde{x} = u}} g^{2}(u)\\
&=& \sum_{u \in V(G_{0})} g^{2}(u)\left( \sum_{vu \in E(G_0)} \sum_{xy \in \mathcal{A}_{ uv }} 1 + 2|\mathcal{A}_{u}| \right) \\
&=& 3^{t} \sum_{u \in V(G_{0})} g^{2}(u)\deg_0(u) + 2(3^{t} - 2^{t})
\sum_{u \in V(G_0)} g^{2}(u)
\\ &=& 3^{t} + 2(3^t - 2^t) \sum_{u \in G_0} g^2(u).
\end{eqnarray*}

By Lemma~$\ref{lemm}$ and proceeding as above, noting that $\sum_{v
\in V(G_0)} g(v) \deg_0(v) = 0$, we have that
\begin{eqnarray*}
\bar{d}^2\vol(G_t) &=& \frac{\left( \sum\limits_{x \in V(G_t)} f(x) \deg_t(x) \right)^{2}}{\vol(G_t)} \\
&=& \frac{ \left(2 (3^{t} - 2^{t}) \sum\limits_{u \in V(G_0)} g(u)\right)^{2} }{\vol(G_t)} \\
&=& \frac{ 4\cdot 3^{2t}\left(1 -
\left(\frac{2}{3}\right)^{t}\right)^{2} \left( \sum\limits_{u \in
V(G_0)} g(u) \right)^{2} }{
3^{t}\left(\vol(G_0) + 2n_0 \left(1 - \left(\frac{2}{3}\right)^{t}\right)\right)} \\
&\leq&  \frac{ 4\cdot 3^{t}\left(1 -
\left(\frac{2}{3}\right)^t\right)^{2}  \sum\limits_{u \in V(G_0)}
g^2(u) }{\bar{D} + 2\left(1 - \left(\frac{2}{3}\right)^t\right)},
\end{eqnarray*}
where $\bar{D}$ is the average degree of $G_0$, and the last
inequality follows from the Cauchy-Schwarz inequality.

By $(\ref{eqn:ral3})$ we have that
\begin{eqnarray*}
\lambda_1(t) &\leq& \frac{\sum\limits_{x y\in E(G_t)} (f(x) - f(y))^{2}}{\sum\limits_{x\in V(G_t)} f^{2}(x) \deg_t(x) + \bar{d}^{2} \vol(G_t)} \\
&\leq& \frac{ 3^{t}\sum\limits_{uv \in E(G_0)} (g(u)-g(v))^2}{3^{t}
+ 2 \cdot 3^t\left(1  - \left(\frac{2}{3}\right)^t\right)\left(
\sum_{u \in V(G_0)} g^2(u)\right)  -
\frac{ 4\cdot 3^{t}\left(1 - \left(\frac{2}{3}\right)^t\right)^{2}  \sum\limits_{u \in V(G_0)} g^2(u) } {\bar{D} + 2\left(1 - \left(\frac{2}{3}\right)^t\right)} }\\
& = & \frac{\lambda_1(0)}{1 + 2\left(1 -
\left(\frac{2}{3}\right)^{t}\right) \left(\sum\limits_{u\in V(G_0)}
g^2(u) \right)\left(1 - \frac{2\left(1 -
\left(\frac{2}{3}\right)^{t}\right)}{\bar{D} +
2\left(1 - \left(\frac{2}{3}\right)^{t}\right) }\right)}\\
& < & \lambda_1(0),
\end{eqnarray*}
where the strict inequality follows from the fact that $\bar{D} \geq
1$ since $G_0$ is connected and $G_0 \ncong K_1$. \qed

\vspace{.1in}
\noindent \textbf{Proof of Theorem~\ref{adj1}.}
We denote vectors in \textbf{bold}. We first assume that $\rho \neq -1$. Hence, $\rho_{+}, \rho_{-} \neq 0 .$
Let $\mathbf{u}$ be an eigenvector of $A=A(G_t)$ such that $A \mathbf{u} =\rho \mathbf{u}.$ Let $\beta=\frac{(\rho + 1)}{\rho},$ and let
$$ {\mathbf v}=
 \left(
\begin{array}{c}
{\mathbf u} \\
\beta{\mathbf u}
\end{array}
 \right).$$
Then we have that
\begin{eqnarray*}
M \mathbf{v} &=& \left(
\begin{array}{cc}
A & A+I \\
A+I & 0%
\end{array}
\right)
\left(
\begin{array}{c}
\mathbf{u} \\
\beta \mathbf{u} %
\end{array}
\right) \\
&=& \left(
\begin{array}{c}
\rho \mathbf{u} + (\rho + 1)\beta \mathbf{u} \\
(\rho + 1)\mathbf {u}
\end{array} \right). \\
\end{eqnarray*}
Now $\beta \rho =\rho + 1$, and so $(\rho+1)\mathbf{u}=\beta \rho \mathbf{u}.$ The condition
$$\rho=\rho +\beta(\rho +1)= \rho +\frac{(\rho + 1)^{2}}{\rho}$$
is equivalent to $\rho$ solving
$$x-\rho-\frac{(\rho +1)^{2}}{x}=0.$$
Hence, $M \mathbf{v} =\rho \mathbf{v}$ as desired.

Now let $\rho=-1$. In this case, $\rho_{-}=-1$. Let
$$ \mathbf{v}= \left(
\begin{array}{c}
\mathbf{u} \\
\mathbf{0}
\end{array} \right),$$
where $\mathbf{0}$ is the appropriately sized zero vector.
Thus,
\begin{eqnarray*}
M \mathbf{v} &=& \left(
\begin{array}{cc}
A & A+I \\
A+I & 0%
\end{array}
\right)
\left(
\begin{array}{c}
\mathbf{u} \\
\mathbf{0} %
\end{array}
\right) \\
&=& \left(
\begin{array}{c}
-\mathbf{u} \\
\mathbf{0} \\
\end{array} \right). \\
\end{eqnarray*}
Hence, $M \mathbf{v}=\rho_{-} \mathbf{v}$ as desired.
In this case where $\rho_{+}=0$ and $\rho=-1$, let
$$ \mathbf{v}= \left(
\begin{array}{c}
\mathbf{0} \\
\mathbf{u}
\end{array} \right),$$
and so $M \mathbf{v}=\rho_{+} \mathbf{v}.$ \qed

\vspace{.1in}
\noindent \textbf{Proof of Theorem~\ref{adj}.} Without loss of
generality, we assume that $G_0$ is not the trivial graph $K_1;$ otherwise, $G_1$ is
$K_2,$ and we may start from there. Thus, in particular, we can
assume $\rho_0(0) \geq 1$.

We first observe that by Theorem~\ref{adj1}
\begin{equation*}
\rho_{0}(t) \geq \left(\frac{1 + \sqrt{5}}{2}\right)^t \rho_0(0).
\label{eqn:triv}
\end{equation*}
By Theorem~\ref{adj1} and
by taking a branch of descendants from the largest eigenvalue it follows that
\[
|\rho_1(t)| \geq \frac{2(\sqrt{5}-1)}{(1 + \sqrt{5})^2}
\left(\frac{1 + \sqrt{5}}{2} \right)^{t} \rho_0(0).
\]
Hence, to prove the theorem, it suffices to show that
\[
\rho_0(t) \leq c\left(\frac{1 + \sqrt{5}}{2}\right)^{t} \rho_0(0).
\]
Observe that, also by Theorem~\ref{adj1} and taking the largest
branch of descendants from the largest eigenvalues,
\[
\rho_0(t) = \rho_0(0)\prod_{i=0}^{t-1} \left(\frac{1 + \sqrt{5 +
\frac{8}{\rho_0(i)} + \frac{4}{\rho^2_0(i)}}}{2}\right) \leq
\rho_0(0) \prod_{i=0}^{t-1} \left(\frac{1 + \sqrt{5 +
\frac{6}{\rho_0(i)}}}{2}\right).
\]
Thus,
\begin{eqnarray*}
\frac{2^t\rho_0(t)}{(1 + \sqrt{5})^{t}} &\leq&
\rho_0(0) \prod_{i=0}^{t-1} \frac{1 + \sqrt{5 + \frac{6}{\rho_0(i)}}}{1 + \sqrt{5}} \\
&\leq& \rho_0(0)\prod_{i=0}^{t-1} \left(1 + \frac{\sqrt{5}}{1 + \sqrt{5}} \frac{6}{5\rho_0(i)} \right) \\
&\leq& \rho_0(0)\exp\left( \frac{6\sqrt{5}}{5(1 + \sqrt{5})} \sum_{i=0}^{t-1} \rho_0(i)^{-1} \right)\\
&\leq& \rho_0(0) \exp\left( \frac{6\sqrt{5}}{5(1 +
\sqrt{5})\rho_0(0)} \sum_{i=0}^{\infty} \left(\frac{2}{1 +
\sqrt{5}}\right)^{-i} \right) = \rho_0(0)c.
\end{eqnarray*}
In all we have proved that for constants $c$ and $d$ that
$$
c \left(\frac{1 + \sqrt{5}}{2}\right)^{t} \rho_0(0) \geq \rho_0(t)
\geq |\rho_1(t)| \geq d\left(\frac{1+\sqrt{5}}{2}\right)^{t}
\rho_0(t).\qed \qedhere
$$

\subsection{Proofs of Theorems~\ref{dom} and \ref{embed}}

We give the proofs for the results on the cop number, domination number, and automorphism group of the ILT model. \vspace{.1in}

\noindent \textbf{Proof of Theorem~\ref{dom}.} We prove that for $t \ge 0$, $\gamma(G_{t+1})= \gamma(G_t)$. It then follows that $\gamma(G_t)=\gamma(G_0)$.
When a dominating node $x \in V(G_t)$ is cloned, its clone $x'$ will be dominated by $x$. The clone $y'$ of a non-dominating node $y \in V(G_t)$ will be joined to a dominating node since $y$ is joined to one. Hence, a dominating set in $G_t$ is a dominating set in $G_{t+1},$ and so $\gamma(G_{t+1})\le \gamma(G_t).$ If $S'$ is a dominating set in $G_{t+1},$ then form $S$ by replacing (if necessary) nodes $x' \in S'$ by nodes $x.$ As $S$ dominates $G_t,$ it follows that $\gamma(G_{t})\le \gamma(G_{t+1}).$

We next show that $c(G_{t+1})=c(G_t)$.
Let $c=c(G_t)$. Assume that $c$ cops play in $G_{t+1}$ so that whenever $\mathcal{R}$ is on $x' \in V(G_{t+1}) \setminus V(G_t)$, the cops $\mathcal{C}$ play as if he were on $x \in V(G_t)$. Either $\mathcal{C}$ captures $\mathcal{R}$ on $x'$, or using their winning strategy in $G_t$, the cops move to $x$ with $\mathcal{R}$ on $x'$. The cops then win in the next round. Hence, $$c(G_{t+1}) \le c(G_t).$$
If $b=c(G_{t+1}) < c,$ then we prove that $c(G_t) \le b,$ which is a contradiction. Suppose that $\mathcal{R}$ and $\mathcal{C}$ play in $G_t.$ At the same time this game is played, let the set of $b$ cops $\mathcal{C}'$ play with their winning strategy in $G_{t+1},$ under the assumption that $\mathcal{R}$ remains in $G_t.$  Each time a cop in $\mathcal{C}'$ moves to a cloned node $x'$, move the corresponding cop in $\mathcal{C}$ to $x.$ As $x$ and $x'$ are joined and share the exact same neighbours in $G_{t+1},$ $\mathcal{C}$ may win in $G_t$ with $b< c$ cops. \qed

\vspace{.1in}

\noindent \textbf{Proof of Theorem~\ref{embed}.}  We first prove the following lemma.

\begin{lemma}\label{ll}
Each $f_0 \in \mathrm{Aut}(G_0)$, extends to $f_t \in \mathrm{Aut}(G_t)$.
\end{lemma}

\noindent \textbf{Proof.} Given $f_0 \in \mathrm{Aut}(G_0)$, we prove by induction on $t \ge 0$ that $f_0$ extends to $f_t \in \mathrm{Aut}(G_t)$. The base case is immediate. Assuming that $f_t$ is defined, let
\begin{equation*}
f_{t+1}(x)= \begin{cases}
f_{t}(x) & \textrm{if $x \in V(G_t)$},\\
(f_{t}(y))' & \textrm{where $x=y'$ }.\\
\end{cases}
\end{equation*}
Let $x,y$ be distinct nodes of $V(G_t)$. It is straightforward to see that $f_{t+1}$ is a bijection. We show that $xy \in E(G_{t+1})$ if and only if $f_{t+1}(x)f_{t+1}(y) \in E(G_{t+1})$. This will prove that $f_{t+1} \in \mathrm{Aut}(G_t)$, as $f_{t+1}$ extends $f_t$.

The case for $x,y \in V(G_t)$ is immediate as $f_{t} \in \mathrm{Aut}(G_t)$. Next, we consider the case for $x \in V(G_t)$ and $y' \in V(G_{t+1})$. Now $xy' \in E(G_{t+1})$ if and only if $$f_{t+1}(x)f_{t+1}(y')=f_{t}(x)(f_{t}(y))' \in E(G_{t+1}).$$
Note that $x'y' \notin E(G_{t+1})$ for all $x',y' \in V(G_{t+1}) \setminus V(G_t)$. But $f_{t+1}(x') f_{t+1}(y') \notin E(G_{t+1})$ by definition of $G_{t+1}.$ \qed
\vspace{.1in}

We now prove that for all $t \ge 0$, $\mathrm{Aut}(G_t)$ is isomorphic to a subgroup of $\mathrm{Aut}(G_{t+1})$. The proof of Theorem~\ref{embed} then follows from this fact by induction on $t$. Define $$ \phi :\mathrm{Aut}(G_t) \rightarrow \mathrm{Aut}(G_{t+1}) $$ by
\begin{equation*}
\phi (f) (x) = \left\{ \begin{array}{ll}
f(x) & $ if $ x \in V(G_t), \\
(f(y))' & $ if $ x=y' \in V(G_{t+1}) \setminus V(G_t).  \end{array}\right.
\end{equation*}
Note that $\phi(f)(x)$ is injective, since $f \neq g$ implies that $ \phi(f) \neq \phi (g)$ by the definition of $\phi$.

We prove that for all $x \in V(G_{t+1})$ and $f,g \in \mathrm{Aut}(G_t)$, $$\phi (fg)(x)=\phi(f) \phi(g) (x).$$
If $x \in V(G_t),$ then $$\phi(fg)(x)=fg(x)=\phi(f) \phi(g) (x).$$ If $x \notin V(G_t)$, then say $x=y'$, with $y \in V(G_t)$.
We then have that
\begin{eqnarray*}
 \phi(fg) (x)
&=& (fg(y))' \\
&=& (\phi(f) \phi(g) (y) )' \\
&=& \phi (f) (g(y))' \\
&=& \phi (f) \phi(g) (x) . \qedhere \qed
\end{eqnarray*}

\subsection{Proofs of Theorems~\ref{random_dens} and \ref{random_gap}}

We give the proofs for the results on the randomized ILT model, ILT($p)$. Without loss of generality, we assume that $0<p<1$. \vspace{.1in}

\noindent \textbf{Proof of Theorem~\ref{random_dens}.} By the definition of the ILT($p$) model, we obtain the following conditional expectation:
\begin{equation*}
\E(\vol(H_{t+1}) ~|~ \vol(H_{t})) = 3 \vol(H_{t}) + n_{t+1} + n_t (n_t - 1) p(n_t).
\end{equation*}

At the beginning of the process, we cannot control the random variable $\vol(H_t)$; it may be far from its expectation. However, if $t$ is large enough, a number of additional edges added in a random process may be controlled, and $\vol(H_t)$ eventually approaches its expected value. Let
\begin{equation}\label{eq:t0}
t_0(T) = \frac {4 \log \log T}{\log(3+\delta)}
\end{equation}
be the time from which we can control the process (note that $t_0(T)$ tends to infinity with $T$). Now suppose that $$\vol(H_{t_0}) = (3+\delta)^{t_0} (1+A(t_0)).$$ The function $A(t_0)$ measures how far $\vol(G_{t_0})$ is from its expectation; we do not give a explicit formula for this, but the bounds $-1 \le A(t_0) \le \left( \frac{4}{3+\delta} \right)^{t_0}$ apply (deterministically; note that $-1$ corresponds to an empty graph, while $\left( \frac{4}{3+\delta} \right)^{t_0}$ corresponds to a complete graph).  We first demonstrate that for any $t$ (where $t_0(T) \le t \le T$) with probability at least $(1-T^{-2})^t$
\begin{equation}\label{eq:claim}
\vol (H_t) = (1+o(1))(3+\delta)^{t} \left(1+\left( \frac {3}{3+\delta} \right)^{t-t_0} A(t_0) \right) .
\end{equation}

We prove (\ref{eq:claim}) by induction on $t$. The base case, where $t = t_0$, trivially holds.
For the inductive step, assume that~(\ref{eq:claim}) holds for $t_0 = t_0(T) \le t < T$ (with probability at least $(1-T^{-2})^t$). We want to show that (\ref{eq:claim}) holds for $t+1$ (with probability at least $(1-T^{-2})^{t+1}$). Using~(\ref{eq:t0}) and~(\ref{eq:p}) we have that the expected number of \emph{random edges} added at time $t+1$ (that is, edges added between new nodes) is
\begin{eqnarray*}
\E X &=& 2^{t} (2^{t}-1) p(2^{t}) \\
&= & (1-(1/2)^t) \delta (3+\delta)^t  \\
&\ge& (1+o(1))\delta (3+\delta)^{t_0}  \\
&\ge& (1+o(1))\delta \log^4 T .
\end{eqnarray*}
Using the Chernoff bound
$$\Prob( |X-\E X| \ge \eps \E X) \le 2 \exp (-\eps^2 \E X / 3)$$
with $\eps = 1/\log T,$ we derive that the number of random edges is not concentrated with probability at most
$$
2 \exp \left(-\frac{\eps^2 \E X}{3} \right) \le 2 \exp \left(-\frac{\delta \log^2 T }{4} \right) \le T^{-2}.
$$

Thus, with probability at least $(1-T^{-2})^{t+1}$ we have that
\begin{eqnarray*}
\vol(H_{t+1}) &=& 3 \vol(H_t) + 2^{t+1} +  (1+O(\log^{-1} T)) \delta (3+\delta)^{t}  \\
&=& (1+o(1))(3+\delta)^{t} \left(3+3\left( \frac {3}{3+\delta} \right)^{t-t_0} A(t_0) + \delta \right)  \\
&=& (1+o(1))(3+\delta)^{t+1} \left(1+\left( \frac {3}{3+\delta} \right)^{t+1-t_0} A(t_0) \right).
\end{eqnarray*}

By the bounds on $A(t_0)$ it follows that
\begin{eqnarray*}
\left( \frac {3}{3+\delta} \right)^{T-t_0} A(t_0) & = &\exp \left( - \Omega(T) + O(t_0) \right)  \\
&= &o(1).
\end{eqnarray*}
Therefore, the assertion holds with probability at least
\begin{eqnarray*}
(1+T^{-2})^T &= &\exp( (1+o(1)) T^{-1}) \\
&=& 1+o(1). \qedhere \qed
\end{eqnarray*}

\noindent \textbf{Proof of Theorem~\ref{random_gap}.}
Let
$$X = V(H_T) \setminus V(H_{T-1})$$ and $$\bar{X} = V(H_T) \setminus X = V(H_{T-1}).$$ By computation it follows that a.a.s.
\begin{eqnarray*}
\vol(X) &=& (1+o(1))(1+\delta) (3+\delta)^{T-1} , \\
\vol(\bar{X}) &=& (1+o(1))2 (3+\delta)^{T-1},  \\
\vol(H_T) &=& (1+o(1))(3+\delta) (3+\delta)^{T-1},
\end{eqnarray*}
and
$$e(X,X) = (1+o(1))(3+\delta)^{T-1}.$$
Thus, by Lemma~\ref{mix} we have that a.a.s.
\begin{eqnarray*}
\lambda(T) &\ge& (1+o(1))\frac {|3+\delta - (1+\delta)^2|}{2 (1+\delta)} \\
& = &(1+o(1))\frac {2 - \delta - \delta^2}{2 (1+\delta)} \\
&=& \Omega(1). \qedhere \qed
\end{eqnarray*}

\section{Conclusion and further work}

We introduced the ILT model for OSNs and other complex
networks, where the network is cloned at each time-step. We proved
that the ILT model generates graphs with a densification power law,
in many cases decreasing average distance (and in all cases, the average distance and
diameter are bounded above by constants independent of time), have
higher clustering than random graphs with the same average degree,
and have smaller spectral gaps for both their normalized Laplacian
and adjacency matrices than in random graphs. The cop and domination number were shown to remain the same
as the graph from the initial time-step $G_0$, and the automorphism group
of $G_0$ is a subgroup of the automorphism group of graphs generated at all later times. A randomized version of the ILT
model was introduced with tuneable densification power law exponent.

As we noted after the statement
of Lemma~\ref{degree}, the ILT model does not generate graphs with a
power law degree distribution, and neither does the ILT($p$) model. An interesting problem is to design and analyze a
randomized version of the ILT model satisfying the properties
displayed in the ILT model as well as generating power
law graphs. Such a randomized ILT model should with high probability generate power law graphs with topological and spectral
properties similar to graphs from the deterministic ILT model.

Certain OSNs like Twitter are directed networks, where users may be either friends with other users (represented by undirected edges), or follow them (represented by a directed edge pointing to the follower). Hence, a more accurate model for such networks would be directed, and we will consider a directed version of the ILT model in the sequel.

\end{document}